\pdfoutput=1
\documentclass[manuscript, nonacm]{acmart}
\AtBeginDocument{%
  \providecommand\BibTeX{{%
    \normalfont B\kern-0.5em{\scshape i\kern-0.25em b}\kern-0.8em\TeX}}}

\setcopyright{acmcopyright}
\copyrightyear{2018}
\acmYear{2018}
\acmDOI{10.1145/1122445.1122456}

\acmConference[Woodstock '18]{Woodstock '18: ACM Symposium on Neural
  Gaze Detection}{June 03--05, 2018}{Woodstock, NY}
\acmBooktitle{Woodstock '18: ACM Symposium on Neural Gaze Detection,
  June 03--05, 2018, Woodstock, NY}
\acmPrice{15.00}
\acmISBN{978-1-4503-XXXX-X/18/06}

\usepackage{listings}
\usepackage{amsmath}
\usepackage{comment}
\usepackage[skins, breakable]{tcolorbox}
\usepackage{multirow}
\usepackage{amsthm}
\usepackage{amsfonts}
\usepackage{natbib}
\usepackage{graphicx}
\usepackage{appendix}
\usepackage{xspace}
\usepackage{color}
\usepackage{tikz}
\usepackage{pgfplots}
\usepackage{textcomp}
\usepackage{enumitem}
\usepackage{url}
\usepackage{subcaption}
\usepackage[T1]{fontenc}
\usepackage{fullpage}
\usepackage{changepage}
\usepackage{wrapfig}
\usepackage{tabularx}
\usepackage{hhline}
\usepackage{algorithm}
\usepackage[noend]{algpseudocode}

\newtheorem{lemma}{Lemma}
\newtheorem{theorem}{Theorem}

\newcommand{\fc} {\mathcal{C}}

\newcommand{\overbar}[1]{\mkern 1.5mu\overline{\mkern-1.5mu#1\mkern-1.5mu}\mkern 1.5mu}

\newcommand{\timeout}{\texttt{timeout}\xspace}

\newcommand{\hash} {H}
\newcommand{\hqc} {qc_{high}}
\newcommand{\rank} {rank}
\newcommand{\fblock} {\overbar{B}}

\newcommand{\fqc} {\overbar{qc}}
\newcommand{\ftc} {\overbar{tc}}
\newcommand{\coinqc} {qc_{coin}}
\newcommand{\fvotedround} {\overbar{r}_{vote}}
\newcommand{\fvotedheight} {\overbar{h}_{vote}}
\newcommand{\mode} {\texttt{fallback-mode}}
\newcommand\sig[1]{\langle #1 \rangle\xspace}
\newcommand\thsig[1]{\{ #1 \}\xspace}
\newcommand\block[1]{[ #1 ]\xspace}

\algnewcommand{\LineComment}[1]{\Statex \hskip\ALG@thistlm {\color{gray}\textrm{// #1}}}

\newtcolorbox{mybox}[1][]{
enhanced,
colback=white,
boxsep=0pt,
#1
}

\renewcommand{\paragraph}[1]{\smallskip\noindent{\textbf{#1}}}
\newcommand\red[1]{\textcolor{red}{#1}\xspace}

\begin{document}

\title{Be Prepared When Network Goes Bad: An Asynchronous View-Change Protocol}

\author{Rati Gelashvili}
\affiliation{%
  \institution{Novi}
  \country{USA}
}

\author{Lefteris Kokoris-Kogias}
\affiliation{%
  \institution{Novi \& IST Austria}
  \country{USA}
}

\author{Alexander Spiegelman}
\affiliation{%
  \institution{Novi}
  \country{USA}
}

\author{Zhuolun Xiang}
\authornote{Work done at Novi.}
\affiliation{%
  \institution{University of Illinois at Urbana-Champaign}
  \country{USA}
}
\email{xiangzl@illinois.com}



\begin{abstract}
The popularity of permissioned blockchain systems demands BFT SMR protocols that are efficient under good network conditions (synchrony) and robust under bad network conditions (asynchrony).
The state-of-the-art partially synchronous BFT SMR protocols provide optimal linear communication cost per decision under synchrony and good leaders, but lose liveness under asynchrony.
On the other hand, the state-of-the-art asynchronous BFT SMR protocols are live even under asynchrony, but always pay quadratic cost even under synchrony.
In this paper, we propose a BFT SMR protocol that achieves the best of both worlds --
optimal linear cost per decision under good networks and leaders,
optimal quadratic cost per decision under bad networks, 
and remains always live.
\end{abstract}

\ccsdesc[500]{Theory of computation~Distributed algorithms}
\ccsdesc[300]{Security and privacy~Distributed systems security}

\keywords{Byzantine faults, state machine replication, asynchrony, view-change protocol, optimal}

\maketitle

\section{Introduction}

The popularity of blockchain protocols has generated a surge in researching how to increase their efficiency and robustness. On increasing the efficiency front, Hotstuff~\cite{yin2019hotstuff} has managed to achieve linear communication complexity, hitting the limit for deterministic consensus protocols~\cite{spiegelman2020search}. This linearity, however, is only provided during good network conditions and when the leader is honest. Otherwise a quadratic pacemaker protocol, described in a production version of HotStuff named DiemBFT~\cite{baudet2019state}, helps slow parties catch-up. In the worse case of an asynchronous adversary there is no liveness guarantee.

For robustness, randomized asynchronous protocols VABA~\cite{abraham2019asymptotically}, Dumbo~\cite{lu2020dumbo} and ACE~\cite{spiegelman2019ace} are the state of the art. The main idea of these protocols is to make progress as if every node is the leader and retroactively decide on a leader. As a result, the network adversary only has a small probability of guessing which nodes to corrupt. To achieve this, however, these protocols pay a quadratic communication cost even under good network conditions.

In this paper we achieve the best of both worlds. When the network behaves well and the consensus makes progress we run DiemBFT~\cite{baudet2019state}, 
adopting its linear communication complexity. However, when there is a problem and the majority of nodes cannot reach the leader we proactively run an asynchronous view-change protocol that makes progress even under the strongest network adversary. As a result, depending on the conditions we pay the appropriate cost while at the same time always staying live.

In Table~\ref{table:comparison} we show a comparison of our work with the state of the art.
In addition to the related work HotStuff, DiemBFT, VABA and ACE, a recent work~\cite{spiegelman2020search} proposes a single-shot validated BA that has quadratic cost under asynchrony, with a synchronous fast path of cost $O(nf')$ under $f'$ actual faults, and $O(n)$ when there are no faults.
Inspired by~\cite{spiegelman2020search}, we propose a protocol for BFT state machine replication that achieves the best of both worlds, and keeps the amortized efficiency and simplicity of the chain-based BFT SMR protocols such as HotStuff and DiemBFT. 
More related work can be found in Appendix~\ref{sec:relatedwork}. Additionally we show in the Appendix~\ref{app:2-chain} how to decrease the latency of our consensus protocol from 3-chain to 2-chain, leveraging the quadratic cost that the pacemaker anyway imposes during bad network conditions.

\begin{table}[h]
\centering
\setlength\doublerulesep{0.5pt}
\begin{tabular}{|c|c|c|c|}
\hline
                 & {\bf Problem}        & {\bf Comm. Complexity} & {\bf Liveness}                          \\ \hhline{====}
HotStuff~\cite{yin2019hotstuff}/Diem~\cite{baudet2019state}    & partially sync SMR    &  sync $O(n)$  & not live if async              \\ \hline
VABA~\cite{abraham2019asymptotically}/Dumbo~\cite{lu2020dumbo}  & async BA     & $O(n^2)$  & always live                                   \\ \hline

ACE~\cite{spiegelman2019ace}              & async SMR & $O(n^2)$ & always live                                   \\ \hline
Spiegelman'20~\cite{spiegelman2020search}    & async BA with fast sync path     & sync $O(nf')$, async $O(n^2)$  & always live          \\ \hhline{====}
Ours             & async SMR with fast sync path    &  sync $O(n)$, async $O(n^2)$ & always live \\ \hline
\end{tabular}
\caption{Comparison to Related Works. For HotStuff/Diem and our protocol, sync $O(n)$ assumes honest leaders.}
\vspace{-4em}

\label{table:comparison}
\end{table}

\section{Preliminaries}
We consider a permissioned system that consists of an adversary and $n$ replicas numbered $1, 2, \ldots, n$, where each replica has a public key certified by a public-key infrastructure (PKI).
Some replicas are Byzantine with arbitrary behaviors and controlled by an adversary; the rest of the replicas are called honest.
The replicas have all-to-all reliable and authenticated communication channels controlled by the adversary.
An execution of a protocol is 
{\em synchronous} if all message delays between honest replicas are bounded by $\Delta$;
is {\em asynchronous} if they are unbounded;
and is {\em partially synchronous} if there is a global stabilization time (GST) after which they are bounded by $\Delta$.
Without loss of generality, we let $n=3f+1$ where $f$ denotes the assumed upper bound on the number of Byzantine faults, which is the optimal worst-case resilience bound for asynchrony, partial synchrony~\cite{dwork1988consensus}, or asynchronous protocols with fast synchronous path~\cite{blum2020network}.

\paragraph{Cryptographic primitives and assumptions.}
We assume standard digital signature and public-key infrastructure (PKI), and use $\sig{m}_i$ to denote a message $m$ signed by replica $i$.
We also assume a threshold signature scheme, where a set of signature shares for message $m$ from $t$ (the threshold) distinct replicas can be combined into one threshold signature of the same length for $m$.
We use $\thsig{m}_i$ to denote a threshold signature share of a message $m$ signed by replica $i$.
We also assume a collision-resistant cryptographic hash function $\hash(\cdot)$ that can map an input of arbitrary size to an output of fixed size. 
We assume a strongly adaptive adversary that is computationally bounded.
For simplicity, we assume the above cryptographic schemes are ideal and a trusted dealer equips replicas with the above cryptographic schemes. The dealer assumption can be lifted using the protocol of Kokoris-Kogias et al.~\cite{kokoris2020asynchronous}.

\paragraph{BFT SMR.}
A Byzantine fault tolerant state machine replication protocol~\cite{Abraham2020SyncHS} commits client transactions as a log akin to a single non-faulty server, and provides the following two guarantees:
\begin{itemize}[leftmargin=*,noitemsep,topsep=0pt]
    \item Safety. Honest replicas do not commit different transactions at the same log position. 
    \item Liveness. Each client transaction is eventually committed by all honest replicas.
\end{itemize}

\smallskip
Besides the two requirements above, a validated BFT SMR protocol also requires to satisfy the {\em external validity}~\cite{cachin2001secure}, which requires any committed transactions to be {\em externally valid}, satisfying some application-dependent predicate. This can be done by adding validity checks on the transactions before the replicas proposing or voting, and for brevity we omit the details and focus on the BFT SMR formulation defined above.
For most of the paper, we omit the client from the discussion and only focus on replicas.

\paragraph{Terminologies.} 
We define several terminologies commonly used in the literature~\cite{yin2019hotstuff,baudet2019state}.
\begin{itemize}[leftmargin=*,noitemsep,topsep=0pt]
    \item {\em Round Number and View Number.} 
    The protocol proceeds in rounds $r=1,2,3,...$, and each round $r$ has a designated leader $L_r$ to propose a new block of round $r$.
    Each replica keeps track of the current round number $r_{cur}$.
    In addition to the round number, every replica also keeps track of the current view number $v_{cur}$ of the protocol. 
    The original DiemBFT does not need to introduce the view number, 
    but it is convenient to have it here for the asynchronous view-change protocol later. The view number starts from $0$ and increments by $1$ after each asynchronous fallback (so in original DiemBFT it is always $0$ and can be neglected).
    
    \item {\em Block Format.}
    A block is formatted as $B=\block{id, qc, r, v, txn}$ where $qc$ is the quorum certificate of $B$'s parent block, $r$ is the round number of $B$, $v$ is the view number of $B$, $txn$ is a batch of new transactions, and $id=\hash(qc,r,v,txn)$ is the unique hash digest of $qc,r,v,txn$. For brevity, we will only specify $qc$, $r$ and $v$ for a new block, and $txn,id$ will follow the same definition above.
    We will use $B.x$ to denote the element $x$ of $B$.
    As mentioned, $B.v$ will always be $0$ for any block in the original DiemBFT description.
    
    \item {\em Quorum Certificate.}
    A quorum certificate (QC) of some block $B$ is a threshold signature of a message that includes $B.id, B.r, B.v$, produced by combining the signature shares $\thsig{B.id, B.r, B.v}$ from a quorum of replicas ($n-f=2f+1$ replicas). 
    We say a block is certified if there exists a QC for the block. 
    The round/view number of $qc$ for $B$ is denoted by $qc.r$/$qc.v$, which equals $B.r$/$B.v$.
    A QC or a block of view number $v$ and round number $r$ has rank $\rank=(v, r)$. QCs or blocks are ranked first by the view number and then by the round number, i.e., $qc_1.rank>qc_2.rank$ if $qc_1.v>qc_2.v$, or $qc_1.v=qc_2.v$ and $qc_1.r>qc_2.r$.
    The function $\max(\rank_1,\rank_2)$ returns the higher rank between $\rank_1$ and $\rank_2$, and
    $\max(qc_1,qc_2)$ returns the higher ranked QC between $qc_1$ and $qc_2$.
    As mentioned, in the original DiemBFT, the view number is always $0$, thus QCs or blocks are ranked only by round numbers.
    
    \item {\em Timeout Certificate.}
    A timeout certificate (TC) is a threshold signature on a round number $r$, produced by combining the signature shares $\thsig{r}$ from a quorum of replicas ($n-f=2f+1$ replicas). 
    
\end{itemize}

\subsection{DiemBFT}

\begin{figure}[t!]
    \centering
    \input{diembft-protocol}
    \caption{DiemBFT Protocol.}
    \vspace{-2em}

    \label{fig:diembft}
\end{figure}

\paragraph{Description of the DiemBFT Protocol.}
The DiemBFT protocol (also known as LibraBFT)~\cite{baudet2019state} is a production version of HotStuff~\cite{yin2019hotstuff} with a synchronizer implementation (Pacemaker), as shown in Figure~\ref{fig:diembft}.
There are two components of DiemBFT, a {\em Steady State protocol} that aims to make progress when the round leader is honest, and a {\em Pacemaker (view-change) protocol} that advances the round number either due to the lack of progress or the current round being completed. 
The leader $L_r$, once entering the round $r$, will propose a block $B$ extending the block certified by the highest QC $\hqc$. 
When receiving the first valid round-$r$ block from $L_r$, any replica tries to advance its current round number, update its highest locked rank and its highest QC, and checks if any block can be committed.
A block can be committed if it is the first block among $3$ adjacent certified blocks with consecutive round numbers. 
After the above steps, the replica will vote for $B$ by sending a threshold signature share to the next leader $L_{r+1}$, if the voting rules are satisfied.
Then, when the next leader $L_{r+1}$ receives $2f+1$ votes above that form a QC of round $r$, it enters round $r+1$ and proposes the block for that round, and the above steps repeat.
When the timer of some round $r$ expires, the replica stops voting for that round and multicast a timeout message containing a threshold signature share for $r$ and its highest QC.
When any replica receives $2f+1$ such timeout messages that form a TC of round $r$, it enters round $r+1$ and sends the TC to the leader $L_{r+1}$.
For space limitation, we omit the correctness proof of the DiemBFT protocol, which can be found in~\cite{baudet2019state}.

The DiemBFT protocol has linear communication complexity per round under synchrony and honest leaders due to the leader-to-all communication pattern and the use of threshold signature scheme\footnote{Although in the implementation, DiemBFT does not use threshold signature, in this paper we consider a version of DiemBFT that uses the threshold signature scheme.}, and quadratic communication complexity for synchronization caused by asynchrony or bad leaders due to the all-to-all multicast of timeout messages.
Note that during the periods of asynchrony, the protocol has {\em no liveness guarantees} -- the leaders may always suffer from network delays, replicas keep multicasting timeout messages and advancing round numbers, and no block can be certified or committed.
To improve the liveness of the DiemBFT protocol and similarly other partially synchronous BFT protocols, we propose a new view-change protocol in the next section named the {\em Asynchronous Fallback protocol}, which has the optimal quadratic communication complexity and always makes progress even under asynchrony.

\section{An Asynchronous View-change Protocol}
\label{sec:fallback}

To strengthen the liveness guarantees of existing partially synchronous BFT protocols such as DiemBFT~\cite{baudet2019state}, we propose an asynchronous view-change protocol (also called asynchronous fallback~\cite{spiegelman2020search}) that has quadratic communication complexity and always makes progress even under asynchrony.
The Asynchronous Fallback protocol is presented in Figure~\ref{fig:fallback_protocol}, which can replace the Pacemaker protocol in DiemBFT (Figure~\ref{fig:diembft}) to obtain a BFT protocol that has linear communication cost for the synchronous path, quadratic cost for the asynchronous path, and is always live.

\paragraph{Additional Terminologies.}
\begin{itemize}[leftmargin=*,noitemsep,topsep=0pt]

    \item {\em Fallback-block and Fallback-chain.}
    For the Asynchronous Fallback protocol, we define another type of block named {\em fallback-block (f-block)}, denoted as $\fblock$. In contrast, the block defined earlier is called the {\em regular block}.
    An f-block $\fblock$ adds two additional fields to a regular block, formatted as $\fblock=\block{B, height, proposer}$ where $B$ is a regular block, $height \in \{1,2,3\}$ is the position of the f-block in the fallback-chain and $proposer$ is the replica that proposes the block.
    We will use $\fblock_{h,i}$ to denote a height-$h$ f-block proposed by replica $i$.
    A fallback-chain (f-chain) consists of f-blocks.
    In the fallback protocol, every replica will construct its fallback-chain (f-chain) with f-blocks, extending the block certified by its $\hqc$.
    
    \item {\em Fallback-QC.}
    A fallback quorum certificate (f-QC) $\fqc$ for an f-block $\fblock_{h,i}$ is a threshold signature for the message $( \fblock.id, \fblock.r, \fblock.v, h, i)$, produced by combining the signature shares $\thsig{\fblock.id, \fblock.r, \fblock.v, h, i}$ from a quorum of replicas ($n-f=2f+1$ replicas). 
    An f-block is certified if there exists an f-QC for the f-block.
    f-QCs or f-blocks are first ranked by view numbers and then by round numbers.
    
    \item {\em Fallback-TC.}
    A fallback timeout certificate (f-TC) $\ftc$ is a threshold signature for a view number $v$, produced by combining the signature shares $\thsig{v}$ from a quorum of replicas ($n-f=2f+1$ replicas). 
    f-TCs are ranked by view numbers. 

    \item {\em Leader Election and Coin-QC.}
    We assume a black-box leader election primitive such as~\cite{loss2018combining}, that can randomly elect a unique leader $L$ among $n$ replicas with probability $1/n$ for any view $v$, by combining $f+1$ valid coin shares of view $v$ from any $f+1$ distinct replicas.
    A coin-QC $\coinqc$ of view $v$ is formed with these $f+1$ coin shares, each produced by a replica.
    The probability of the adversary to predict the outcome of the election is at most $1/n$ (we assume the cryptographic schemes are ideal). 

    \item {\em Endorsed Fallback-QC and Endorsed Fallback-block.}
    Once a replica has a coin-QC $\coinqc$ of view $v$ that elects replica $L$ as the leader, we say any f-QC of view $v$ by replica $L$ is endorsed (by $\coinqc$), and the f-block certified by the f-QC is also endorsed (by $\coinqc$).
    Any endorsed f-QC is handled as a QC in any steps of the protocol such as {\em Lock, Commit, Advance Round}.
    {\em An endorsed f-QC ranked higher than any QC have with the same view number.}
    As cryptographic evidence of endorsement, the first block in a new view can additionally include the coin-QC of the previous view, which is omitted from the protocol description for brevity.
    
\end{itemize}

\begin{figure}[t]
    \centering
    \input{sys-protocol}
    \caption{DiemBFT with Asynchronous Fallback.}
    \vspace{-1em}
    \label{fig:fallback_protocol}
\end{figure}

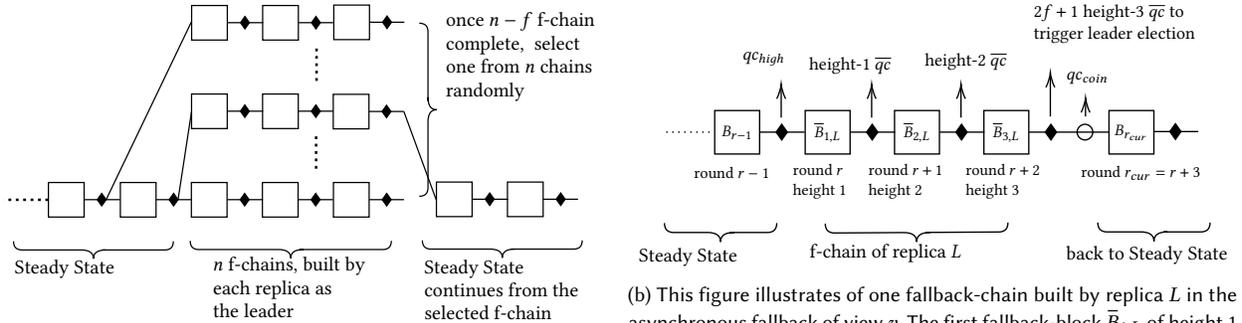
\begin{figure}[t]
     \centering
     \begin{subfigure}[b]{0.49\textwidth}
         \centering
         \resizebox{\textwidth}{!}{
        \tikzset{every picture/.style={line width=0.75pt}} 

\begin{tikzpicture}[x=0.75pt,y=0.75pt,yscale=-1,xscale=1]

\draw   (122.62,163.37) -- (151.28,163.37) -- (151.28,191.13) -- (122.62,191.13) -- cycle ;
\draw   (179.66,163.37) -- (208.32,163.37) -- (208.32,191.13) -- (179.66,191.13) -- cycle ;
\draw   (236.7,163.37) -- (265.36,163.37) -- (265.36,191.13) -- (236.7,191.13) -- cycle ;
\draw  [fill={rgb, 255:red, 0; green, 0; blue, 0 }  ,fill opacity=1 ] (165.4,172.16) -- (168.96,177.15) -- (165.4,182.14) -- (161.83,177.15) -- cycle ;
\draw    (151.49,177.15) -- (180.01,177.15) ;
\draw  [fill={rgb, 255:red, 0; green, 0; blue, 0 }  ,fill opacity=1 ] (222.44,172.16) -- (226,177.15) -- (222.44,182.14) -- (218.87,177.15) -- cycle ;
\draw    (208.53,177.15) -- (237.05,177.15) ;
\draw  [fill={rgb, 255:red, 0; green, 0; blue, 0 }  ,fill opacity=1 ] (279.48,172.16) -- (283.04,177.15) -- (279.48,182.14) -- (275.91,177.15) -- cycle ;
\draw    (265.57,177.15) -- (293.38,177.15) ;
\draw   (64.86,163.37) -- (93.53,163.37) -- (93.53,191.13) -- (64.86,191.13) -- cycle ;
\draw  [fill={rgb, 255:red, 0; green, 0; blue, 0 }  ,fill opacity=1 ] (107.64,172.16) -- (111.21,177.15) -- (107.64,182.14) -- (104.08,177.15) -- cycle ;
\draw    (93.74,177.15) -- (122.26,177.15) ;
\draw   (179.87,21.12) -- (208.53,21.12) -- (208.53,48.88) -- (179.87,48.88) -- cycle ;
\draw   (236.91,21.12) -- (265.57,21.12) -- (265.57,48.88) -- (236.91,48.88) -- cycle ;
\draw    (111.21,177.15) -- (179.59,34.48) ;
\draw  [fill={rgb, 255:red, 0; green, 0; blue, 0 }  ,fill opacity=1 ] (222.65,29.92) -- (226.22,34.91) -- (222.65,39.9) -- (219.09,34.91) -- cycle ;
\draw    (208.75,34.91) -- (237.27,34.91) ;
\draw  [fill={rgb, 255:red, 0; green, 0; blue, 0 }  ,fill opacity=1 ] (279.69,29.92) -- (283.26,34.91) -- (279.69,39.9) -- (276.13,34.91) -- cycle ;
\draw    (265.79,34.91) -- (293.59,34.91) ;
\draw   (179.8,92.26) -- (208.46,92.26) -- (208.46,120.02) -- (179.8,120.02) -- cycle ;
\draw   (236.84,92.26) -- (265.5,92.26) -- (265.5,120.02) -- (236.84,120.02) -- cycle ;
\draw    (179.8,106.31) -- (168.96,177.15) ;
\draw  [fill={rgb, 255:red, 0; green, 0; blue, 0 }  ,fill opacity=1 ] (222.58,101.05) -- (226.14,106.04) -- (222.58,111.03) -- (219.01,106.04) -- cycle ;
\draw    (208.68,106.04) -- (237.2,106.04) ;
\draw  [fill={rgb, 255:red, 0; green, 0; blue, 0 }  ,fill opacity=1 ] (279.62,101.05) -- (283.19,106.04) -- (279.62,111.03) -- (276.06,106.04) -- cycle ;
\draw    (265.72,106.04) -- (293.52,106.04) ;
\draw [line width=1.5]  [dash pattern={on 1.69pt off 2.76pt}]  (280.19,55.69) -- (280.19,81.89) ;
\draw   (293.74,163.37) -- (322.4,163.37) -- (322.4,191.13) -- (293.74,191.13) -- cycle ;
\draw  [fill={rgb, 255:red, 0; green, 0; blue, 0 }  ,fill opacity=1 ] (336.52,172.16) -- (340.08,177.15) -- (336.52,182.14) -- (332.95,177.15) -- cycle ;
\draw    (322.61,177.15) -- (350.42,177.15) ;
\draw   (293.59,92.49) -- (322.26,92.49) -- (322.26,120.25) -- (293.59,120.25) -- cycle ;
\draw  [fill={rgb, 255:red, 0; green, 0; blue, 0 }  ,fill opacity=1 ] (336.38,101.29) -- (339.94,106.28) -- (336.38,111.27) -- (332.81,106.28) -- cycle ;
\draw    (322.47,106.28) -- (350.28,106.28) ;
\draw   (294.02,20.77) -- (322.69,20.77) -- (322.69,48.53) -- (294.02,48.53) -- cycle ;
\draw  [fill={rgb, 255:red, 0; green, 0; blue, 0 }  ,fill opacity=1 ] (336.8,29.56) -- (340.37,34.55) -- (336.8,39.54) -- (333.24,34.55) -- cycle ;
\draw    (322.9,34.55) -- (350.71,34.55) ;
\draw   (178.73,210.24) .. controls (178.73,214.91) and (181.06,217.24) .. (185.73,217.24) -- (249.48,217.24) .. controls (256.15,217.24) and (259.48,219.57) .. (259.48,224.24) .. controls (259.48,219.57) and (262.81,217.24) .. (269.48,217.24)(266.48,217.24) -- (333.23,217.24) .. controls (337.9,217.24) and (340.23,214.91) .. (340.23,210.24) ;
\draw   (362.68,174.41) .. controls (367.35,174.41) and (369.68,172.08) .. (369.68,167.41) -- (369.68,114.89) .. controls (369.68,108.22) and (372.01,104.89) .. (376.68,104.89) .. controls (372.01,104.89) and (369.68,101.56) .. (369.68,94.89)(369.68,97.89) -- (369.68,42.37) .. controls (369.68,37.7) and (367.35,35.37) .. (362.68,35.37) ;
\draw   (36.84,210.59) .. controls (36.84,215.26) and (39.17,217.59) .. (43.84,217.59) -- (91.01,217.59) .. controls (97.68,217.59) and (101.01,219.92) .. (101.01,224.59) .. controls (101.01,219.92) and (104.34,217.59) .. (111.01,217.59)(108.01,217.59) -- (158.18,217.59) .. controls (162.85,217.59) and (165.18,215.26) .. (165.18,210.59) ;
\draw [line width=1.5]  [dash pattern={on 1.69pt off 2.76pt}]  (64.65,177.26) -- (29,177.26) ;
\draw   (376.45,163.01) -- (405.11,163.01) -- (405.11,190.77) -- (376.45,190.77) -- cycle ;
\draw   (433.49,163.01) -- (462.15,163.01) -- (462.15,190.77) -- (433.49,190.77) -- cycle ;
\draw  [fill={rgb, 255:red, 0; green, 0; blue, 0 }  ,fill opacity=1 ] (419.23,171.8) -- (422.79,176.8) -- (419.23,181.79) -- (415.66,176.8) -- cycle ;
\draw    (405.32,176.8) -- (433.84,176.8) ;
\draw  [fill={rgb, 255:red, 0; green, 0; blue, 0 }  ,fill opacity=1 ] (476.27,171.8) -- (479.83,176.8) -- (476.27,181.79) -- (472.7,176.8) -- cycle ;
\draw    (462.36,176.8) -- (490.17,176.8) ;
\draw    (350.28,106.28) -- (376.23,177.15) ;
\draw   (365.18,209.88) .. controls (365.18,214.55) and (367.51,216.88) .. (372.18,216.88) -- (419.35,216.88) .. controls (426.02,216.88) and (429.35,219.21) .. (429.35,223.88) .. controls (429.35,219.21) and (432.68,216.88) .. (439.35,216.88)(436.35,216.88) -- (486.52,216.88) .. controls (491.19,216.88) and (493.52,214.55) .. (493.52,209.88) ;
\draw [line width=1.5]  [dash pattern={on 1.69pt off 2.76pt}]  (279.83,126.28) -- (279.83,152.48) ;

\draw (195.64,221.95) node [anchor=north west][inner sep=0.75pt]  [font=\LARGE] [align=left] {$\displaystyle n$ f-chains, built by\\ each replica as \\ the leader};
\draw (36.45,223.82) node [anchor=north west][inner sep=0.75pt]  [font=\LARGE] [align=left] {Steady State};
\draw (382.65,25.07) node [anchor=north west][inner sep=0.75pt]  [font=\LARGE] [align=left] {once $\displaystyle n-f$ f-chain \\complete, \ select \\one from $\displaystyle n$ chains\\randomly};
\draw (365.25,223.59) node [anchor=north west][inner sep=0.75pt]  [font=\LARGE] [align=left] {Steady State \\continues from the \\selected f-chain};

\end{tikzpicture}
        }
         \caption{This figure illustrates the high-level picture of the protocol. During the fallback, each replica acts as a leader and tries to build a certified fallback-chain. When $2f+1$ f-chains are completed, randomly elects one leader from $n$ replicas, and all replicas continue the Steady State chain from the f-chain built by the elected leader.}
         \label{fig:illustration1}
     \end{subfigure}
     \hfill
     \begin{subfigure}[b]{0.49\textwidth}
         \centering
         \resizebox{\textwidth}{!}{
        \tikzset{every picture/.style={line width=0.75pt}} 

\begin{tikzpicture}[x=0.75pt,y=0.75pt,yscale=-1,xscale=1]

\draw   (169.96,79.51) -- (203.35,79.51) -- (203.35,111.85) -- (169.96,111.85) -- cycle ;
\draw   (236.4,79.51) -- (269.79,79.51) -- (269.79,111.85) -- (236.4,111.85) -- cycle ;
\draw   (302.85,79.51) -- (336.24,79.51) -- (336.24,111.85) -- (302.85,111.85) -- cycle ;
\draw   (395.88,79.51) -- (429.27,79.51) -- (429.27,111.85) -- (395.88,111.85) -- cycle ;
\draw  [fill={rgb, 255:red, 0; green, 0; blue, 0 }  ,fill opacity=1 ] (219.79,89.75) -- (223.95,95.57) -- (219.79,101.38) -- (215.64,95.57) -- cycle ;
\draw    (203.6,95.57) -- (236.82,95.57) ;
\draw  [fill={rgb, 255:red, 0; green, 0; blue, 0 }  ,fill opacity=1 ] (286.24,89.75) -- (290.39,95.57) -- (286.24,101.38) -- (282.09,95.57) -- cycle ;
\draw    (270.04,95.57) -- (303.27,95.57) ;
\draw  [fill={rgb, 255:red, 0; green, 0; blue, 0 }  ,fill opacity=1 ] (352.69,89.75) -- (356.84,95.57) -- (352.69,101.38) -- (348.53,95.57) -- cycle ;
\draw    (336.49,95.57) -- (395.88,95.57) ;
\draw    (352.27,83.94) -- (352.27,53.34) ;
\draw [shift={(352.27,51.34)}, rotate = 450] [color={rgb, 255:red, 0; green, 0; blue, 0 }  ][line width=0.75]    (10.93,-3.29) .. controls (6.95,-1.4) and (3.31,-0.3) .. (0,0) .. controls (3.31,0.3) and (6.95,1.4) .. (10.93,3.29)   ;
\draw    (285.82,84.77) -- (285.82,61.85) ;
\draw [shift={(285.82,59.85)}, rotate = 450] [color={rgb, 255:red, 0; green, 0; blue, 0 }  ][line width=0.75]    (10.93,-3.29) .. controls (6.95,-1.4) and (3.31,-0.3) .. (0,0) .. controls (3.31,0.3) and (6.95,1.4) .. (10.93,3.29)   ;
\draw    (219.38,84.77) -- (219.38,61.85) ;
\draw [shift={(219.38,59.85)}, rotate = 450] [color={rgb, 255:red, 0; green, 0; blue, 0 }  ][line width=0.75]    (10.93,-3.29) .. controls (6.95,-1.4) and (3.31,-0.3) .. (0,0) .. controls (3.31,0.3) and (6.95,1.4) .. (10.93,3.29)   ;
\draw   (102.68,79.51) -- (136.07,79.51) -- (136.07,111.85) -- (102.68,111.85) -- cycle ;
\draw  [fill={rgb, 255:red, 0; green, 0; blue, 0 }  ,fill opacity=1 ] (152.52,89.75) -- (156.67,95.57) -- (152.52,101.38) -- (148.36,95.57) -- cycle ;
\draw    (136.32,95.57) -- (169.54,95.57) ;
\draw    (152.1,84.77) -- (152.1,61.85) ;
\draw [shift={(152.1,59.85)}, rotate = 450] [color={rgb, 255:red, 0; green, 0; blue, 0 }  ][line width=0.75]    (10.93,-3.29) .. controls (6.95,-1.4) and (3.31,-0.3) .. (0,0) .. controls (3.31,0.3) and (6.95,1.4) .. (10.93,3.29)   ;
\draw   (184.91,163.67) .. controls (184.94,168.34) and (187.28,170.66) .. (191.95,170.63) -- (242.85,170.32) .. controls (249.52,170.28) and (252.86,172.59) .. (252.89,177.26) .. controls (252.86,172.59) and (256.18,170.24) .. (262.85,170.2)(259.85,170.22) -- (313.75,169.89) .. controls (318.42,169.86) and (320.74,167.51) .. (320.71,162.84) ;
\draw   (387.57,162.84) .. controls (387.57,167.51) and (389.9,169.84) .. (394.57,169.84) -- (419.31,169.84) .. controls (425.98,169.84) and (429.31,172.17) .. (429.31,176.84) .. controls (429.31,172.17) and (432.64,169.84) .. (439.31,169.84)(436.31,169.84) -- (464.04,169.84) .. controls (468.71,169.84) and (471.04,167.51) .. (471.04,162.84) ;
\draw   (65.31,163.26) .. controls (65.31,167.93) and (67.64,170.26) .. (72.31,170.26) -- (97.04,170.26) .. controls (103.71,170.26) and (107.04,172.59) .. (107.04,177.26) .. controls (107.04,172.59) and (110.37,170.26) .. (117.04,170.26)(114.04,170.26) -- (141.78,170.26) .. controls (146.45,170.26) and (148.78,167.93) .. (148.78,163.26) ;
\draw [line width=0.75]  [dash pattern={on 0.84pt off 2.51pt}]  (66.55,95.57) -- (102.27,95.57) ;
\draw   (372.2,95.57) .. controls (372.2,92.36) and (374.81,89.75) .. (378.02,89.75) .. controls (381.23,89.75) and (383.83,92.36) .. (383.83,95.57) .. controls (383.83,98.78) and (381.23,101.38) .. (378.02,101.38) .. controls (374.81,101.38) and (372.2,98.78) .. (372.2,95.57) -- cycle ;
\draw    (378.43,83.94) -- (378.43,72.03) ;
\draw [shift={(378.43,70.03)}, rotate = 450] [color={rgb, 255:red, 0; green, 0; blue, 0 }  ][line width=0.75]    (10.93,-3.29) .. controls (6.95,-1.4) and (3.31,-0.3) .. (0,0) .. controls (3.31,0.3) and (6.95,1.4) .. (10.93,3.29)   ;
\draw  [fill={rgb, 255:red, 0; green, 0; blue, 0 }  ,fill opacity=1 ] (445.3,89.34) -- (449.45,95.15) -- (445.3,100.96) -- (441.14,95.15) -- cycle ;
\draw    (429.1,95.15) -- (462.32,95.15) ;

\draw (186.65,95.68) node    {$\overline{B}_{1,L}$};
\draw (412.57,95.68) node    {$B_{r_{cur}}$};
\draw (171.54,37.9) node [anchor=north west][inner sep=0.75pt]  [font=\large] [align=left] {height-$\displaystyle 1$ $\displaystyle \overline{qc}$};
\draw (159.4,117.93) node [anchor=north west][inner sep=0.75pt]   [align=left] {round $\displaystyle r$\\height $\displaystyle 1$};
\draw (119.38,95.68) node    {$B_{r-1}$};
\draw (122.11,34.41) node [anchor=north west][inner sep=0.75pt]  [font=\large] [align=left] {$\displaystyle qc_{high}$};
\draw (85.95,119.79) node [anchor=north west][inner sep=0.75pt]   [align=left] {round $\displaystyle r-1$};
\draw (373.74,119.37) node [anchor=north west][inner sep=0.75pt]   [align=left] {round $\displaystyle r_{cur} =r+3$};
\draw (172,177.73) node [anchor=north west][inner sep=0.75pt]  [font=\Large] [align=left] {f-chain of replica $\displaystyle L$};
\draw (362.73,178.51) node [anchor=north west][inner sep=0.75pt]  [font=\Large] [align=left] {back to Steady State};
\draw (338.4,-1.22) node [anchor=north west][inner sep=0.75pt]  [font=\large] [align=left] {$\displaystyle 2f+1$ height-3 $\displaystyle \overline{qc}$ to \\trigger leader election};
\draw (215.68,117.93) node [anchor=north west][inner sep=0.75pt]   [align=left] {round $\displaystyle r+1$\\height $\displaystyle 2$};
\draw (287.94,117.93) node [anchor=north west][inner sep=0.75pt]   [align=left] {round $\displaystyle r+2$\\height $\displaystyle 3$};
\draw (253.1,95.68) node    {$\overline{B}_{2,L}$};
\draw (319.55,95.68) node    {$\overline{B}_{3,L}$};
\draw (257.48,34.9) node [anchor=north west][inner sep=0.75pt]  [font=\large] [align=left] {height-$\displaystyle 2$ $\displaystyle \overline{qc}$};
\draw (44.99,180.09) node [anchor=north west][inner sep=0.75pt]  [font=\Large] [align=left] {Steady State};
\draw (363.01,49.36) node [anchor=north west][inner sep=0.75pt]  [font=\large] [align=left] {$\displaystyle qc_{coin}$};

\end{tikzpicture}
        }
         \caption{This figure illustrates of one fallback-chain built by replica $L$ in the asynchronous fallback of view $v$. The first fallback-block $\fblock_{1,L}$ of height $1$ contains the $\hqc$, and the second and third f-block contains the f-QC of the previous f-block. $2f+1$ height-$3$ f-QCs will trigger the leader election to form a coin-QC $\coinqc$ that randomly selects a replica as the leader with probability $1/n$. Replicas continue the Steady State from the f-chain by the leader.}
         \label{fig:illustration2}
     \end{subfigure}
        \caption{Illustrations of the Asynchronous Fallback Protocol.}
        \vspace{-1em}
        \label{fig:illustrations}
\end{figure}

\paragraph{Description of the Asynchronous Fallback Protocol.}
The illustration of the fallback protocol can be found in Figure~\ref{fig:illustrations}.
Now we give a description of the Asynchronous Fallback protocol (Figure~\ref{fig:fallback_protocol}), which replaces the Pacemaker protocol in the DiemBFT protocol (Figure~\ref{fig:diembft}).
Each replica additionally keeps a boolean value $\mode$ to record if it is in a fallback, during which the replica will not vote for any regular block. 
When the timer for round $r$ expires, the replica enters the fallback and multicast a timeout message containing its highest QC $\hqc$ and a threshold signature share of the current view number.
When receiving $2f+1$ timeout messages of the same view number no less than the current view, the replica forms a fallback-TC $\ftc$ and enters the fallback.
The replica also updates its current view number $v_{cur}$, and initializes the voted round number $\fvotedround[j]=0$ and the voted height number $\fvotedheight[j]=0$ for each replica $j$. 
Finally, the replica starts to build the fallback-chain by multicasting the f-TC $\ftc$ and proposing the first fallback-block of height $1$, round number $\hqc.r+1$, view number $v_{cur}$, and has the block certified by $\hqc$ as the parent.
When receiving such a height-$1$ f-block proposed by any replica $j$, replicas that are in the fallback and not yet voted for height-$1$ f-block by $j$, will verify the validity of the f-block, update the voted round and height number for the fallback, and vote for the f-block by sending the threshold signature share back to the replica $j$.
When the replica receives $2f+1$ votes that form a fallback-QC for its height-$1$ f-block, the replica can propose the next f-block in its f-chain, which has height and round number increased by $1$, and has the height-$1$ f-block as the parent.
When receiving such a height-$2$ f-block proposed by any replica $j$, replicas that are in the fallback and not yet voted for any height-$2$ f-block by $j$, will similarly update the voted round and height number and vote for the f-block, after verifying the validity of the f-block.
Similarly, when the height-$2$ f-block gets certified, the proposer can propose the height-$3$ f-block extending the height-$2$ f-block.
When the height-$3$ f-block gets certified by an f-QC, the proposer multicasts the height-$3$ f-QC.
When the replica receives $2f+1$ height-$3$ f-QCs of the current view, it knows that $2f+1$ f-chains have been completed.
The replica then will start the leader election by signing and multicasting a coin share for the current view number, and a unique leader of the view can be elected by any of the $f+1$ valid coin share messages.
When receiving such $f+1$ coin share messages, a coin-QC $\coinqc$ can be formed.
The fallback is finished once a $\coinqc$ is received, and the replica updates $\mode=false$ to exit the fallback and enters the next view. To ensure consistency, the replica also updates the highest voted round number $r_{vote}$ to be the voted round number for the leader during the fallback, and updates its highest QC $\hqc$, highest locked rank $\rank_{lock}$ and check for commit as in the {\em Lock} step.
Recall that for any certified f-block proposed by the leader $L$ elected by the coin-QC of view $v$ and the f-QC are endorsed, and will be handled as a QC for the regular block. QCs, f-QCs, blocks and f-blocks are ranked first by the view numbers and then the round numbers, except that endorsed f-blocks (f-QCs) rank higher than certified blocks (QCs) if they have the same view number.
As cryptographic evidence of endorsement, the first block in a new view can additionally include the coin-QC of the previous view, which is omitted from the protocol description for brevity.


\paragraph{Optimization in Practice.}
Instead of always having all replica build their own fallback-chain and wait for $2f+1$ of them to finish, in practice we can optimize the performance of the protocol by allowing the replicas to build their fallback-chain on top of others' fallback-chain.
For instance, once replica $i$ sees the first certified height-$1$ fallback-block by any replica $j$, replica $i$ can propose its height-$2$ f-block extending that height-$1$ f-block instead of waiting for its height-$1$ f-block to be certified.
Any replica $i$ can even entirely adopt other replicas' completed f-chain, once replica $i$ receives all certified f-blocks of that f-chain.
Such the above optimization, the fallback protocol can proceed in the speed of the fastest replica instead of the fastest $2f+1$ replicas, also possibly reduce the communication cost.

\subsection{Proof of Correctness}
\label{sec:proof}

Recall that we say a fallback-block $\fblock$ of view $v$ by replica $i$ is endorsed, if $\fblock$ is certified and there exists a coin-QC of view $v$ that elects replica $i$ as the leader.

We say a replica is in the Steady State of view $v$, if it has $\mode=false$ and $v_{cur}=v$; otherwise if $\mode=true$ and $v_{cur}=v$, we say the replica is in the asynchronous fallback of view $v$.

We use $\fc(B)$ to denote the set of replicas who provide votes for the QC or f-QC for $B$.

If a block or f-block $B$ is an ancestor of another block or f-block $B'$ due to a chain of QCs or endorsed f-QCs, we say $B'$ extends $B$. A block or f-block also extends itself.

\textbf{Rules for Leader Rotation.} Let $\mathcal{L}=\{L_1,L_2,...\}$ be the infinite sequence of predefined leaders corresponding to the infinite round numbers starting from $1$. To ensure the liveness of synchronous fast path, we rotate the leader once every $4$ rounds, i.e., $L_{4k+1},...,L_{4k+4}$ are the same replica for every $k\geq 0$.

\begin{lemma}\label{lem:1}
    Let $B,B'$ both be endorsed f-blocks of the same view, or both be certified blocks of the same view.
    If the round number of $B$ equals the round number of $B'$, then $B=B'$.
\end{lemma}

\begin{proof}
    Suppose on the contrary that $B\neq B'$. Let $r$ be the round number of $B,B'$.
    
    Suppose that $B,B'$ both are certified blocks of the same view.
    Since $n=3f+1$ and $|\fc(B)|=|\fc(B')|=2f+1$, by quorum intersection, $\fc(B)\cap \fc(B')$ contains at least one honest replica $h$ who voted for both $B$ and $B'$.
    Without loss of generality, suppose that $h$ voted for $B$ first.
    According to the {\bf Vote} step, $h$ updates its $r_{vote}=r$ and will only vote for blocks with round number $>r_{vote}$ in the same view. 
    Therefore $h$ will not vote for $B'$, thus $B'$ cannot be certified, contradiction.
    
    Suppose that $B,B'$ both are endorsed f-blocks of the same view.
    By quorum intersection, at least one honest replica $h$ voted for both $B$ and $B'$. Without loss of generality, suppose that $h$ voted for $B$ first.
    Since $B,B'$ are of the same view with elected leader $L$, 
    according to the {\bf Fallback Vote} step, after voting for $B$ of round $r$, $h$ updates $\fvotedround[L]= r$ and will only vote for blocks of round number $>\fvotedround[L]$.
    Thus $h$ will not vote for $B'$ and $B'$ cannot be certified, contradiction.

\end{proof}

\begin{lemma}\label{lem:2}
    For any chain that consists of only certified blocks and endorsed f-blocks, the adjacent blocks in the chain have consecutive round numbers, and nondecreasing view numbers. Moreover, for blocks of the same view number, no endorsed f-block can be the parent of any certified regular block.
\end{lemma}

\begin{proof}
    Suppose on the contrary that there exist adjacent blocks $B,B'$ of round number $r,r'$ where $B$ is the parent block of $B'$, and $r'\neq r+1$.
    If $B'$ is a certified block, according to the {\bf Vote} step, honest replicas will not vote for $B'$ since $r'\neq r+1$, and $B'$ cannot be certified, contradiction.
    If $B$ is an endorsed f-block, according to the {\bf Fallback Vote} step, honest replicas will not vote for $B'$ since $r'\neq r+1$, and $B'$ cannot be certified, contradiction.
    Therefore, the blocks in the chain have consecutive round numbers.
    
    Suppose on the contrary that there exist adjacent blocks $B,B'$ of view number $v,v'$ where $B$ is the parent block of $B'$, and $v'<v$.
    Since $n=3f+1$ and $|\fc(B)|=|\fc(B')|=2f+1$, by quorum intersection, $\fc(B)\cap \fc(B')$ contains at least one honest replica $h$ who voted for both $B$ and $B'$.
    According to the {\bf Vote} and {\bf Fallback Vote} steps, $h$ has its $v_{cur}=v$ when voting for $B$.
    According to the protocol, $h$ only updates its $v_{cur}$ in steps {\bf Enter Fallback} and {\bf Exit Fallback}, and $v_{cur}$ is nondecreasing.
    Hence, $h$ will not vote for $B'$ since $v'<v_{cur}$, contradiction.
    Therefore, the blocks in the chain have nondecreasing view numbers.
    
    Suppose on the contrary that there exist adjacent blocks $B,B'$ of the same view number $v$ where $B$ is the parent block of $B'$, and $B$ is an endorsed f-block and $B'$ is a certified block.
    Since $n=3f+1$ and $|\fc(B)|=|\fc(B')|=2f+1$, by quorum intersection, $\fc(B)\cap \fc(B')$ contains at least one honest replica $h$ who voted for both $B$ and $B'$.
    According to the {\bf Fallback Vote} step, $h$ has $\mode=true$ when voting for $B$. 
    According to the {\bf Vote} step, $h$ has $\mode=false$ when voting for $B'$.
    The only step for $h$ to set its $\mode$ to $false$ is the {\bf Exit Fallback} step, where $h$ also updates its current view to $v+1$. Then, $h$ will not vote for the view-$v$ block $B'$, contradiction.
    Therefore, for blocks of the same view number, no fallback-block can be the parent of any regular block.
\end{proof}

\begin{lemma}\label{lem:3}
    Let $B,B'$ both be endorsed f-blocks of the same view, then either $B$ extends $B'$ or $B'$ extends $B$.
\end{lemma}

\begin{proof}
    Suppose on the contrary that $B,B'$ do not extend one another, and they have height numbers $h,h'$ respectively. Let $L$ be the replica that proposes $B,B'$ and elected by the coin-QC.
    Without loss of generality, assume that $h\leq h'$.
    According to the {\bf Fallback Vote} step, each honest replica votes for f-blocks of $L$ with strictly increasing height numbers. 
    Since by quorum intersection, there exists at least one honest replica $h$ that voted for both $B$ and $B'$, thus we have $h<h'$.
    Since $h'>h$, and according to the {\bf Fallback Vote} step, there must be another endorsed f-block $B''$ by $L$ of height $h$ that $B'$ extends, and $B,B''$ do not extend one another. Similarly, by quorum intersection, at least one honest replica voted for both $B$ and $B''$, which is impossible since each honest replica votes for f-blocks of $L$ with strictly increasing height numbers. 
    Therefore, $B,B'$ extend one another.
\end{proof}

\begin{lemma}\label{lem:4}
    If there exist three adjacent certified or endorsed blocks $B_r,B_{r+1},B_{r+2}$ in the chain with consecutive round numbers $r,r+1,r+2$ and the same view number $v$, then any certified or endorsed block of view number $v$ that ranks no lower than $B_r$ must extend $B_r$.
\end{lemma}

\begin{proof}
    Suppose on the contrary that there exists a block of view number $v$ that ranks no lower than $B_r$ and does not extend $B_r$. Let $B$ be such a block with the smallest rank, then the parent block $B'$ of $B$ also does not extend $B_r$, but ranks no higher than $B_r$.
    By Lemma~\ref{lem:2}, there are $3$ cases: (1) $B_r,B_{r+2}$ are certified blocks, (2) $B_r$ is a certified block and $B_{r+2}$ is an endorsed blocks, and (3) $B_r,B_{r+2}$ are endorsed blocks.
    
    Suppose $B_r,B_{r+2}$ are certified blocks.
    By quorum intersection, $\fc(B)\cap \fc(B_{r+2})$ contains at least one honest replica $h$ who voted for both $B$ and $B_{r+2}$.
    According to the {\bf Vote} and {\bf Lock} step, when $h$ votes for $B_{r+2}$, it sets its $\rank_{lock}$ to be the rank of $B_r$.
    Suppose that $B$ is a certified block, by Lemma~\ref{lem:1}, $B$ must has round number $\geq r+3$ otherwise $B$ will extend $B_r$.
    If the parent block $B'$ of $B$ has view number $<v$, then $B'$ ranks lower than $B_r$. 
    If $B'$ has view number $v$, by Lemma~\ref{lem:2} and \ref{lem:1}, $B'$ is a certified block and has round number $<r$, thus also ranks lower than $B_r$.
    According to the {\bf Vote} step, $h$ will not vote for $B$ since $B'$ has rank lower than its $\rank_{lock}$ in round $\geq r+3$, contradiction.
    Suppose that $B$ is a endorsed block, let $B_1$ be the height-$1$ endorsed block of view $v$, and let $B_1'$ be the parent block of $B_1$, which ranks no higher than $B_r$.
    If $B_1'$ is a certified block, by the same argument above we have $B_1'$ ranks lower than $B_r$.
    If $B_1'$ is an endorsed block, since $B_1'$ ranks no higher than $B_r$ its has view number $\leq v$. By definition an endorsed block ranks higher than a certified block if they have the same view number, thus $B_1'$ must have view number $<v$ and rank lower than $B_r$.
    By quorum intersection, at least one honest replica $h'$ voted for both $B_1$ and $B_{r+2}$, and has $\rank_{lock}$ to be the rank of $B_r$.
    However, according to the {\bf Fallback Vote} step, $h'$ will not vote for $B_1$ since $B_1'$ ranks lower than its $\rank_{lock}$, contradiction.
    
    Suppose $B_r$ is a certified block and $B_{r+2}$ is an endorsed blocks.
    Suppose that $B$ is a certified block, by Lemma~\ref{lem:1}, $B$ must has round number $\geq r+1$ otherwise $B$ will extend $B_r$. Since $B,B_r$ does not extend one another, by Lemma~\ref{lem:2}, there exists a round-$r$ block $B_r'\neq B_r$ that $B$ extends. By Lemma~\ref{lem:2}, $B_r'$ has view number $\leq v$. 
    If $B_r'$ has view number $v$, by Lemma~\ref{lem:1}, $B_r=B_r'$, contradiction.
    If $B_r'$ has view number $\leq v-1$, by quorum intersection, there exists at least one honest replica that voted for both $B_r'$ and $B_r$.
    According to the steps {\bf Vote, Fallback Vote}, after voting for $B_r'$, $h$ sets its $r_{vote}=r$ and will not vote for regular blocks of round number $\leq r$. Hence, $h$ will not vote for $B_r$ of round number $r$, contradiction.
    Suppose that $B$ is an endorsed block, by Lemma~\ref{lem:3}, $B$ and $B_{r+2}$ extend one another. By Lemma~\ref{lem:2}, no endorsed f-block can be the parent of any certified block of the same view, therefore $B$ must extend $B_r$, contradiction.

    Suppose $B_r,B_{r+2}$ are endorsed blocks.
    Suppose that $B$ is a certified block. Since $B,B_r$ both have view number $v$,
    and by definition an endorsed block ranks higher than a certified block if they have the same view number, $B$ ranks lower than $B_r$, contradiction.
    Suppose that $B$ is an endorsed block. By Lemma~\ref{lem:3}, $B$ and $B_r$ extend one another. Since $B$ ranks no lower than $B_r$ and endorsed f-blocks in the chain have consecutive round numbers, we have $B$ extends $B_r$, contradiction.
    
    Therefore, any certified or endorsed block of view number $v$ that ranks no lower than $B_r$ must extend $B_r$.
    
\end{proof}

\begin{lemma}\label{lem:5}
    If a block $B_r$ is committed by some honest replica due to a $3$-chain starting from $B_r$, and another block $B_{r'}'$ is committed by some honest replica due to a $3$-chain starting from $B_{r'}'$, then either $B_r$ extends $B_{r'}'$ or $B_{r'}'$ extends $B_r$.
\end{lemma}

\begin{proof}
    Suppose on the contrary that $B_r,B_{r'}'$ are committed by honest replica, but they do not extend one another.
    Suppose there exist three adjacent certified or endorsed blocks $B_r,B_{r+1},B_{r+2}$ in the chain with consecutive round numbers $r,r+1,r+2$ and the same view number $v$. 
    Suppose there also exist three adjacent certified or endorsed blocks $B_{r'},B_{r'+1}',B_{r'+2}'$ in the chain with consecutive round numbers $r',r'+1,r'+2$ and the same view number $v'$. 
    Without loss of generality, suppose that $v\leq v'$.
    
    If $v=v'$, without loss of generality, further assume that $r\leq r'$. 
    By Lemma~\ref{lem:4}, any certified or endorsed block of view number $v$ that ranks no lower than $B_r$ must extend $B_r$. Since $v=v'$ and $r\leq r'$, $B_{r'}'$ ranks no lower than $B_r$, and thus $B_{r'}'$ must extend $B_r$, contradiction.
    
    If $v<v'$, by Lemma~\ref{lem:2}, there exist blocks that $B_{r'}'$ extends, have view numbers $> v$, and do not extend $B_r$.
    Let $B$ be such a block with the smallest view number. 
    Then the parent block $B'$ of $B$ also does not extend $B_r$, and has view number $\leq v$.
    By quorum intersection, at least one honest replica $h$ voted for both $B_{r+2}$ and $B$.
    According to the {\bf Lock} step, after $h$ voted fro $B_{r+2}$ in view $v$, it sets its highest locked rank $\rank_{lock}$ to be the rank of $B_r$.
    Then, when $h$ votes for $B$, according to the steps {\bf Vote, Fallback Vote}, the rank of the parent block $B'$ must be $\geq \rank_{lock}$, and thus no lower than the rank of $B_r$.
    By definition $B'$ has view number $\leq v$, hence $B'$ must have view number $v$.
    By Lemma~\ref{lem:4}, $B'$ of view number $v$ and ranks no lower than $B_r$ must extend $B_r$, contradiction.
    
    Therefore, we have either $B_r$ extends $B_{r'}'$, or $B_r$ extends $B_{r'}'$.
    
\end{proof}

\begin{theorem}[Safety]
    If blocks $B$ and $B'$ are committed at the same height in the blockchain by honest replicas, then $B=B'$.
\end{theorem}

\begin{proof}
    Suppose that $B$ is committed due to a block $B_l$ of round $l$ being directly committed by a $3$-chain, and $B'$ is committed due to a block $B_k$ of round $k$ being directly committed by a $3$-chain. 
    By Lemma~\ref{lem:5}, either $B_l$ extends $B_k$ or $B_l$ extends $B_k$, which implies that $B=B'$.
\end{proof}

\begin{lemma}\label{lem:6}
    If all honest replicas enter the asynchronous fallback of view $v$ by setting $\mode=true$, then eventually they all exit the fallback and set $\mode=false$.
    Moreover, with probability $2/3$, at least one honest replica commits a new block after exiting the fallback.
\end{lemma}

\begin{proof}
    Suppose that all honest replicas set $\mode=true$ and have $v_{cur}=v$.
    
    We first show that all honest replicas will vote for the height-$1$ fallback-block proposed by any honest replica.
    Suppose on the contrary that some honest replica $h$ does not vote for the height-$1$ f-block proposed by some honest replica $h'$.
    According to the step {\bf Enter Fallback, Fallback Vote}, the only possibility is that $qc.\rank< \rank_{lock}$, where $qc$ is the QC contained in the height-$1$ f-block by $h'$ and $\rank_{lock}$ is the highest locked rank of $h$ when voting for the f-block. 
    Suppose that $\rank_{lock}=(v, r)$. 
    According to the protocol, the only step for $h$ to set its $\rank_{lock}$ is the {\bf Lock} step, and $h$ must receive a QC for some block $B$ of round $r+1$ since all adjacent certified blocks have consecutive round numbers by Lemma~\ref{lem:2}
    According to the {\bf Timer and Timeout} step, honest replicas set $\mode\gets false$ when sending the timeout message for view $v$, and will not vote for any regular blocks when $\mode=true$ according to the {\bf Vote} step. 
    Since $B$ of round $r+1$ is certified, by quorum intersection, at least one honest replica votes for $B$ and then sends a timeout message with $\hqc$ of rank $(v,r)$ which is received by $h'$. Then, the highest ranked $qc$ among all timeout messages received by $h'$ should have rank $qc.\rank \geq (v,r)=\rank_{lock}$, contradiction.
    Hence, all honest replicas will vote for the height-$1$ fallback-block proposed by any honest replica.
    
    After the height=$1$ f-block is certified with $2f+1$ votes that forms an f-QC $\fqc_1$, according to the {\bf Fallback Propose} step, replica $i$ can propose a height-$2$ f-block extending the height-$2$. All honest replicas will vote for the height-$2$ according to the {\bf Fallback Vote} step, and thus it will be certified with $2f+1$ votes.
    Then similarly replica $i$ can propose a height-$3$ f-block which will be certified by $2f+1$ votes, and replica $i$ will multicast the f-QC for this height-$3$ f-block.
    
    Since any honest replica $i$ can obtain a height-$3$ f-QC and multicast it, any honest replica will receive $2f+1$ valid height-$3$ f-QCs and then sign and multicast a coin share of view $v$ in the {\bf Leader Election} step. 
    When the $\coinqc$ of view $v$ that elects some replica $L$ is received or formed at any honest replica, at least one honest replica has received $2f+1$ height-$3$ f-QCs before sending the coin share, which means that $2f+1$ fallback-chains have $3$ f-blocks certified.
    All honest replicas will receive the $\coinqc$ eventually due to the forwarding.
    According to the {\bf Exit Fallback} step, they will set $\mode=false$ and exit the fallback.
    
    Since the coin-QC elects any replica to be the leader with probability $1/n$, 
    and at least one honest replica receives $2f+1$ height-$3$ f-QCs among all $3f+1$ f-chains. 
    With probability $2/3$, the honest replica has the height-$3$ f-QC of the f-chain by the elected leader, thus have the $3$ f-blocks endorsed which have the same view number and consecutive round numbers.
    Therefore, the honest replica can commit the height-$1$ f-block according to the {\bf Commit} step.
    
\end{proof}

\begin{theorem}[Liveness]
    All honest replicas keep committing new blocks with high probability.
\end{theorem}

\begin{proof}
    We prove by induction on the view numbers.
    
    We first prove for the base case where all honest replicas start their protocol with view number $v=0$.
    If all honest replicas eventually all enter the asynchronous fallback, by Lemma~\ref{lem:6}, they eventually all exit the fallback, and a new block is committed at least one honest replica with probability $2/3$. According to the {\bf Exit Fallback} step, all honest replicas enter view $v=1$ after exiting the fallback.
    If at least one honest replica never set $\mode=true$, this implies that the sequence of QCs produced in view $0$ is infinite. By Lemma~\ref{lem:2}, the QCs have consecutive round numbers, and thus all honest replicas keep committing new blocks after receiving these QCs.
    
    Now assume the theorem is true for view $v=0,...,k-1$. Consider the case where all honest replicas enter the view $v=k$.
    By the same argument for the $v=0$ base case, honest replicas either all enter the fallback and the next view with a new block committed with $2/3$ probability, or keeps committing new blocks in view $k$. 
    When the network is synchronous and the leaders are honest, the block proposed in the Steady State will always extend the highest QC, and thus voted by all honest replicas.
    Therefore, by induction, honest replicas keep committing new blocks with high probability.
    
\end{proof}

\begin{theorem}[Efficiency]
    During the periods of synchrony with honest leaders, the amortized communication complexity per block decision is $O(n)$.
    During periods of asynchrony, the expected communication complexity per block decision is $O(n^2)$.
\end{theorem}

\begin{proof}
    When the network is synchronous and leaders are honest, no honest replica will multicast timeout messages. In every round, the designated leader multicast its proposal of size $O(1)$ (due to the use of threshold signatures for QC), and all honest replicas send the vote of size $O(1)$ to the next leader.
    Hence the communication cost is $O(n)$ per round and per block decision.
    
    When the network is asynchronous and honest replicas enter the asynchronous fallback, each honest replica in the fallback only broadcast $O(1)$ number of messages, and each message has size $O(1)$. Hence, each instance of the asynchronous fallback has communication cost $O(n^2)$, and will commit a new block with probability $2/3$.
    Therefore, the expected communication complexity per block decision is $O(n^2)$.
\end{proof}

\section{Get $2$-chain Commit for Free}\label{app:2-chain}

In this section, we show how to slightly modify the $3$-chain commit protocol in Figure~\ref{fig:fallback_protocol} into a $2$-chain commit protocol, strictly improving the commit latency without losing any of the nice guarantees.
Earlier solutions that can be adapted into the chain-based BFT SMR framework with $2$-chain commit either pay a quadratic cost for the view-change~\cite{castro1999practical, jalalzai2020fast}, or sacrifices responsiveness (cannot proceed with network speed)~\cite{buchman2016tendermint, buterin2017casper, yin2019hotstuff}.
Since our protocol already pays the quadratic cost, which is inevitable for the asynchronous view-change~\cite{abraham2019asymptotically}, we can get the improvement over the commit latency of the protocol from $3$-chain commit ($6$ rounds) to $2$-chain commit ($4$ rounds), without sacrificing anything else.

\begin{figure}[h!]
    \centering
    \input{2-chain-protocol}
    \caption{DiemBFT with Asychronous Fallback and $2$-chain Commit. Differences Marked in \red{red}.}
    \label{fig:2-chain}
\end{figure}

The difference compared with the protocol in Figure~\ref{fig:fallback_protocol} with $3$-chain commit is marked red in Figure~\ref{fig:2-chain}.
Intuitively, $2$-chain commit means the protocol only needs two adjacent certified blocks in the same view to commit the first block, and to ensure safety, now replicas lock on the highest QC's rank ($1$-chain lock) instead of the parent QC's rank of the highest QC.
Moreover, since now a block can be committed with $2$-chain, the asynchronous fallback only needs every replica to build a fallback-chain of $2$ fallback-blocks instead of $3$, to ensure progress made during the fallback.
Hence, $2$ rounds can be reduced during the fallback if the protocol uses a $2$-chain commit.
Another main difference is that, with $2$-chain commit and $1$-chain lock, only one honest replica $h$ may have the highest QC among all honest replicas when entering the asynchronous fallback. Then, only the fallback-chain proposed by $h$ will get $2f+1$ votes from all honest replica and be completed, while other f-chains proposed by other honest replica extending a lower-ranked QC may not get $2f+1$ votes, since $h$ will not vote due to the $1$-chain locking.
A straightforward solution is to allow replicas to adopt f-chains from other replicas, and build their f-chains on top of the first certified f-chain they know.
With the changes above, even if only one fallback-chain is live, all honest replicas can complete their fallback-chain and ensure the liveness of the asynchronous fallback.

To summarize, the $2$-chain-commit version of the protocol strictly improves the latency of the $3$-chain commit version, by reducing the commit latency by $2$ rounds for both Steady State and Asynchronous Fallback.
The correctness proof of the $2$-chain commit protocol would be analogous to that for the $3$-chain commit version, and we omit it here for brevity.

\section{Related Works}
\label{sec:relatedwork}

\paragraph{Byzantine fault tolerant replication.}
BFT SMR has been studied extensively in the literature.
For partial synchrony, PBFT is the first practical protocol with $O(n^2)$ communication cost per decision under a stable leader and $O(n^3)$ communication cost for view-change.
A sequence of efforts\cite{kotla2007zyzzyva,abraham2018revisiting,buchman2016tendermint,buterin2017casper,gueta2019sbft,yin2019hotstuff} have been made to reduce the communication cost of the BFT SMR protocols, with the state-of-the-art being HotStuff~\cite{yin2019hotstuff} that has $O(n)$ cost for decisions and view-changes under synchrony and honest leaders, and $O(n^2)$ cost view-changes otherwise.
For asynchrony, several recent proposals focus on improving the communication complexity and latency, including HoneyBadgerBFT~\cite{miller2016honey}, Dumbo-BFT~\cite{guo2020dumbo}, VABA~\cite{abraham2019asymptotically}, Dumbo-MVBA~\cite{lu2020dumbo}, ACE~\cite{spiegelman2019ace} and Aleph~\cite{gkagol2019aleph}. 
The state-of-the-art protocols for asynchronous SMR have $O(n^2)$ cost per decision~\cite{spiegelman2019ace}.
There are also recent works on BFT SMR under synchrony~\cite{hanke2018dfinity,Abraham2020SyncHS,abraham2020optimal,nibesh2020optimality}.

\paragraph{Flexibility in BFT protocols.}
A related line of work investigates synchronous BFT protocol with an asynchronous fallback, focusing on achieving optimal resilience tradeoffs between synchrony and asynchrony~\cite{blum2019synchronous,blum2020network} or optimal communication complexity~\cite{spiegelman2020search}.
FBFT~\cite{malkhi2019flexible} proposes one BFT SMR solution supporting clients with various fault and synchronicity beliefs, but the guarantees only hold for the clients with the correct assumptions.
Another related line of work on optimistic BFT protocols focuses on including a fast path in BFT protocols to improve the performance such as communication complexity and latency under the optimistic scenarios, for asynchronous protocols~\cite{kursawe2005optimistic, kursawe2002optimistic}, partially synchronous protocols~\cite{abd2005fault, kotla2007zyzzyva, guerraoui2010next} and synchronous protocols~\cite{pass2018thunderella, cryptoeprint:2018:980, nibesh2020optimality, cryptoeprint:2020:406}.

\section{Conclusion}

We present an asynchronous view-change protocol that improves the liveness of the state-of-the-art partially synchronous BFT SMR protocol.
As a result, we obtain a BFT SMR protocol that has linear communication cost under synchrony, quadratic communication cost under asynchrony, and remains always live.

\bibliographystyle{ACM-Reference-Format}
\bibliography{references}

\appendix


\end{document}